\documentclass{article}

\usepackage{amsmath}
\usepackage{amssymb}
\usepackage{theorem}

\usepackage{enumerate}
\usepackage{microtype}
\usepackage{framed}

\newcommand{\bra}{\langle}
\newcommand{\ket}{\rangle}

\newcommand{\fall}[1]{{\forall\,{#1},\ }}
\newcommand{\fexist}[1]{{\exists\,{#1}\,{:}\ }}

\newcommand{\mc}[1]{{\mathcal{#1}}}

\newcommand{\mb}[1]{{\bf #1}}

\newcommand{\ortho}{\mathrel{\bot}}

\newcommand{\bbot}{{\bot\!\!\!\bot}}

\newcommand{\Card}{\mathop{\mathrm{Card}}}

\newcommand{\cor}{\text{ or }}
\newcommand{\cand}{\text{ and }}

\newtheorem{theorem}{Theorem}
\newtheorem{definition}{Definition}
\newtheorem{proposition}{Proposition}
\newenvironment{proof}[1][Proof]{\paragraph{{#1}}}%
                {{\hfill\(\Box\)\\}}

\makeatletter
\newcommand\At{\mathop{\operator@font At}\nolimits}
\makeatother

\begin{document}

\title{Orthogonality and Dimensionality}
\author{Olivier Brunet\footnote{\texttt{olibrunet at normalesup.org}}}
\maketitle

\begin{abstract}
In this article, we present what we believe to be a simple way to motivate the use of Hilbert spaces in quantum mechanics. To achieve this, we study the way the notion of dimension can, at a very primitive level, be defined as the cardinality of a maximal collection of mutually orthogonal elements (which, for instance, can be seen as spatial directions). Following this idea, we develop a formalism based on two basic ingredients, namely an orthogonality relation and matroids which are a very generic algebraic structure permitting to define a notion of dimension.

Having obtained what we call \emph{orthomatroids}, we then show that, in high enough dimension, the basic ingredients of orthomatroids (more precisely the simple and irreducible ones) are isomorphic to generalized Hilbert lattices, so that the latter are a direct consequence of an orthogonality-based characterization of dimension.
\end{abstract}

\section{Introduction}

One of the most striking peculiarities of the mathematical formulation of quantum mechanics is its heavy reliance on the use of complex Hilbert spaces. In its core formulation, the state of a quantum system is represented by a normalized vector of a Hilbert space and the probabilistic nature of quantum mechanics is formalized by interpreting, using the Born rule, the quantity $\bra \psi | P_i | \psi \ket$ as the probability of obtaining eigenvalue $\lambda_i$ (with $P_i$ denoting the orthogonal projection on the associated eigenspace) when measuring a quantum system in state $|\psi\ket$, thus emphasizing the importance of the inner product.

However, fundamental as it is, it seems that there still exists no consensus regarding any satisfactory justification of the use of Hilbert spaces, so that one might wonder how far we stand from Mackey which had to state, in a rather \emph{ad hoc} manner, its seventh axiom \cite{Mackey57Book}~as \emph{``the partially ordered set (or poset) of all questions in quantum mechanics is isomorphic to the poset of all closed subspaces of a [...
] Hilbert space.''}
This requirement can usually be split into two parts, and many efforts have tried to address both of them: first that the poset of all questions forms an orthomodular lattice, and that this orthomodular lattice is indeed isomorphic to the lattice of all closed subspaces of a Hilbert space.

Regarding orthomodularity, described by Beltrametti and Cassinelli \cite{Beltrametti81:LogicQM} as \emph{``the survival [...] of a notion of the logical conditional, which takes the place of the classical implication associated with Boolean algebra''}, some attempts to justify this as a property verified by the poset of all questions include works by Grinbaum \cite{Grinbaum2005FPL} based on the notion of \emph{``relevant information''}
and by the author of the present article \cite{Brunet04IQSA} where it is assumed that the considered lattice ensures the definition of sufficiently many \emph{``points of view.''}

As for the second requirement, a central result, Piron's celebrated Representation Theorem, provides conditions for an orthomodular lattice to be isomorphic to the lattice of all closed subspaces of a Hilbert space (or, more precisely, to a \emph{generalized} Hilbert space, that is where it is not required that the used field is a ``classical'' one) \cite{piron64,Piron76Book,Stubbe:2007}. However, one of those conditions, namely the \emph{covering law}, \emph{``presents a [...] delicate problem. [...] It is probably safe to say that no simple and entirely compelling argument has been given for assuming its general validity''}~\cite{Wilce_QuantLog_Plato}.

\ 

In the present article, we present some simple geometric conditions which are sufficient to ensure that the hypotheses of Piron's Representation Theorem are verified. Our approach is based on the notion of dimension and, more precisely, on the fact that this very notion can indeed be defined. Actually, it is taken for granted to the point that this requirement, if not implicite, often appears at the very beginning of axiomatic formulations of physics and, more specifically, of quantum mechanics. For instance, in \cite{Hardy-Axioms}, Hardy assumes that one can define a \emph{number of degrees of freedom} $K$ before presenting the \emph{five reasonable axioms} of the article. In Rovelli's Relational Quantum Mechanics \cite{Rovelli96RQM}, Postulate 1 states that
\emph{``there is a maximum amount of relevant information that can be extracted from a system,''}
which can be interpreted as the existence of the dimension of a system.

Usually, mathematical formulations of physics involve the use of euclidian or Hilbert spaces, where the dimension is defined as the common cardinality of every bases of the considered space. This definition can even be refined by restricting it to orthogonal bases. In other words, the dimension of a euclidian or Hilbert space is the common cardinality of every maximal set of mutually orthogonal vector rays (which we will call a maximal \emph{orthoindependant} set of vector rays). The use of orthogonal bases is a key element on our approach. Intuitively, they constitute a natural set of preferred bases. Mathematically, while the linear independance of the elements of a basis requires to consider the basis as a whole, it can be done in a finitary way for orthogonal bases, since a set is orthoindependant if, and only if, its finite subsets are.

However, while this remark highlights the important role played, in such a context, by orthogonality, it is also clear that this notion can be defined without relying upon the sophisticated machinery of linear and bilinear algebra. Indeed, one can easily imagine a greek philosopher who would, during the Antiquity, justify that the space surrounding him is 3-dimensional by exhibiting three orthogonal rods and by arguing that he cannot add another rod orthogonal to the first three.

To that respect, we believe that orthogonality should be considered as a primitive geometric notion and we will, in this article, study the way it can lead to a reasonable definition of dimension (which, indeed, is also basically a geometric notion). 

\ 

Our primary tool will be the theory of matroids \cite{Kung82:MatroidSourcebook,Oxley:MatroidTheory} which is a very natural algebraic framework for dealing with dimension. This theory is based on a general notion of independence (such as the linear independence of the elements of a basis in linear algebra) which we shall define here using orthogonality. Moreover, following our previous discussion about the role of orthogonal bases, we will add the requirement that orthogonal bases do exist and, more precisely, that any set of mutually orthogonal elements (i.e. any orthoindependant subset) can be completed into an orthogonal basis. This will constitute the next section. We will then show that the obtained algebraic structure, which we call \emph{orthomatroids}, is closely related to Piron's propositional systems, and finally we will present an adaptation of Piron's Representation Theorem to orthomatroids.

\section{Orthomatroids}

Let us first define the notion of orthogonality on a~set.

\begin{definition}
Given a set $E$, a binary relation $\bot$ on $E$ is an \emph{orthogonality relation} if, and only if it verifies
\begin{align*}
& \fall {a, b \in E} a \ortho b \iff b \ortho a & \hbox{Symmetry} \\
& \fall {a, b \in E} a \ortho b \implies a \neq b & \hbox{Anti-reflexivity}
\end{align*}
In the following, such a pair $(E, \bot)$ will be called an \emph{orthoset}.
\end{definition}

The basic intuition behind an orthoset $(E, \bot)$ is to think of its elements as spatial directions. Obviously, the study of such structures is neither new nor original. One might refer for instance to the survey in \cite{Wilce:Handbook}. However, we present some basic results in order to introduce the important ideas before considering orthosets in the light of matroids.

\begin{definition}
Given an orthoset $(E, \bot)$ and a subset $F \in \wp(E)$, we define its \emph{orthogonal complement} $F^\bot$~as
$$ F^\bot = \{ x \in E \mathbin\vert \fall {y \in F} x \ortho y \}, $$
and its (bi-orthogonal) closure $F^{\bbot} = \bigl(F^\bot\bigr)^\bot$.
\end{definition}

\begin{proposition}
Given an orthoset $(E, \bot)$, we~have
\begin{align*}
\fall {F \in \wp(E)} & F \cap F^\bot = \emptyset \\
\fall {F, G \in \wp(E)} & G \subseteq F^\bot \iff F \subseteq G^\bot & \mathrm{(GC)}
\end{align*}
\end{proposition}
\begin{proof}
For $F \in \wp(E)$ and $x \in F$, having $x \in F^\bot$ would imply that $\fall{y \in F} x \ortho y$ and, in particular, $x \ortho x$ which is not possible. As a consequence, $F \cap F^\bot = \emptyset$. Now, (GC) trivially follows from the symmetry of our orthogonality relation, since
$$ G \subseteq F^\bot \iff \fall {x \in F} \fall {y \in G} x \ortho y \iff F \subseteq G^\bot. $$
\end{proof}
This property shows that there is an antitone Galois connection 
\cite{Davey90Book,Erne92GaloisPrimer} between $E$ and itself, realized both ways by ${\ \cdot\ }^\bot : F \mapsto F^\bot$. A direct consequence of this is the following result.
\begin{proposition}
The bi-orthogonal closure operation is indeed a closure operator on $E$, that is it verifies
\begin{align*}
  \fall {F \in \wp(E)} & F \subseteq F^\bbot & \mathrm{Extensivity} \\
  \fall {F, G \in \wp(E)} & F \subseteq G \implies F^\bbot \subseteq G^\bbot \vphantom{\bigr)^{\bot}} & \mathrm{Monotony} \\
  \fall {F \in \wp(E)} & \bigl(F^\bbot\bigr)^\bbot = F^\bbot & \mathrm{Idempotence}
\end{align*}
\end{proposition}
\begin{proof}
As stated above, this is a direct consequence of having a Galois connection. However, the direct proof of those facts is rather easy. Extensivity trivially follows from (GC), since
$$ F \subseteq F^\bbot \iff F^\bot \subseteq F^\bot. $$
Now, if $F \subseteq G$, then $F \subseteq G^\bbot$ and hence $G^\bot \subseteq F^\bot$. By applying this deduction once again, we obtain
$$ F \subseteq G \implies G^\bot \subseteq F^\bot \implies F^\bbot \subseteq G^\bbot. $$ 
Finally, for idempotence, we only need to prove that $\bigl(F^\bbot\bigr)^\bbot \subseteq F^\bbot$ which is equivalent, because of (GC), to $F^\bot \subseteq \bigl((F^\bbot)^\bbot\bigr)^\bot = \bigl((F^\bot)^\bbot\bigr)^\bbot$.
\end{proof}

Following our initial discussion, we now want to turn an orthoset $(E, \bot)$ into a matroid in order to have a suitable definition of dimension. Since we have defined a closure operation on $E$ using our orthogonality relation, it is natural to use the closure operator-based definition of a matroid. In that case, we only need to demand that the Mac\,Lane--Steinitz Exchange Property is verified by $\cdot\ ^\bbot$:
$$ \fall{F \in \wp(E)} \fall {x, y \in E} x \in (F + y)^\bbot \setminus F^\bbot \implies y \in (F + x)^\bbot. \eqno{\mathrm{(EP)}} $$
Here, $F + x$ denotes the set $F \cup \{x\}$. With this definition, a subset $F$ of $E$ is independent if, and only~if (where $F - x$ denotes $F \setminus \{x\}$)
$$ \fall {x \in F} x \not \in (F - x)^\bbot. $$

\ 

The next step is to add the possibility of defining orthobases in a convenient way. To that respect, we first have to focus on subsets made of mutually orthogonal elements.

\begin{definition}
A subset $F$ of $E$ is said to be \emph{orthoindependent} if, and only~if
$$ \fall {x, y \in F} x \neq y \implies x \ortho y. $$
\end{definition}
Obviously, every orthoindependent subset is independent. Moreover, orthoindependent subsets verify the nice following property:
\begin{proposition} \label{prop:zorn}
If $\{F_i\}$ is a chain of orthoindependent subsets, then $\bigcup F_i$ is also orthoindependent.
\end{proposition}

By application of Zorn's Lemma, the previous proposition implies that there exists maximal ortho\-independent subsets and even that every orthoindependent subset is included in a maximal one. 

\ 

As explained in the introduction, we will demand that every closed subset $F^\bbot$ of $E$ admits an orthobasis. We will even demand the stronger condition that every orthoindependent subset $I$ of $F^\bbot$ can be completed into an orthobasis $B$ of $F^\bbot$, which we formally state~as
\begin{quotation}
(OB) Given a subset $F$ of $E$ and an orthoindependent subset $I$ of $F^\bbot$, there exists an orthoindependent subset $B$ such that $I \subseteq B$ and $B^\bbot = F^\bbot$.
\end{quotation}

The requirement that \emph{any} orthoindependent subset can be completed into an orthobasis is just a form of isotropy, stating that there are no ``privileged'' orthoindependent subsets. Moreover, considering Proposition \ref{prop:zorn}, $(OB)$ is equivalent to the statement that every maximal orthoindependent subset is indeed an orthobasis.

In the next proposition, we provide an equivalent and convenient way to state this condition.

\begin{proposition}
If $(E, \bot)$ verifies the MacLane\,--\,Steinitz Exchange Property, then axiom (OB) is equivalent to the \emph{Straightening Property}, which we define~as
$$ \fall {F \in \wp(E)} \fall {x \in E} x \not \in F^\bbot \implies \fexist {y \in F^\bot} x \in (F + y)^\bbot. \eqno{\mathrm{(SP)}}$$
\end{proposition}
\begin{proof}
Suppose first that (OB) holds, and let $F \in \wp(E)$ and $x \in E$ be such that $x \not \in F^\bbot$. Moreover, let $B$ be an orthobasis of $F^\bbot$ which we extend into an orthobasis $B'$ of $(F + x)^\bbot$, using (OB). For $y \in B' \setminus B$, it is clear from the orthoindependence of $B'$ that $y \in F^\bot$. This means that $y \in (F + x)^\bbot \setminus F^\bbot$ so that using the Exchange Property, we get $x \in (F + y)^\bbot$.

Conversely, suppose that the Straightening Property is verified. It can be remarked that, because of the Exchange Property, it can be equivalently stated~as
$$ \fall {F \in \wp(E)} \fall {x \in E} x \not \in F^\bbot \implies \fexist {y \in F^\bot} (F + x)^\bbot = (F + y)^\bbot. $$
Now, given an orthoindependent subset $I$ of $F^\bbot$, let $J$ be a maximal orthoindependent subset of $F^\bbot$ such that $I \subseteq J$. If $J^\bbot \neq F^\bbot$, then there exists an element $x \in F^\bbot \setminus J^\bbot$ and, following from the Straightening Property, there exists an element $y \in J^\bot$ such that $(J + x)^\bbot = (J + y)^\bbot$. This implies in particular that $y \in F^\bbot$ and hence $J + y$ is also an orthoindependent subset of $F^\bbot$ which is absurd since $J$ was supposed maximal. As a consequence, we have $J^\bbot = F^\bbot$, that is $J$ is an orthobasis of $F^{\bbot}$.
\end{proof}

We now summarize all these properties into what we call an \emph{orthomatroid}:

\begin{definition}[Orthomatroid]
An \emph{orthomatroid} is an orthoset $(E, \bot)$ which verifies the two following properties~:
\begin{enumerate}
  \item Exchange Property
  $$ \fall {F \in \wp(E)} \fall {x,y \in E} x \in (F + y)^\bbot \setminus F^\bbot \implies y \in (F + x)^\bbot $$
  \item Straightening Property
  $$ \fall {F \in \wp(E)} \fall {x \in E} x \not \in F^\bbot \implies \fexist {y \in F^\bot} x \in (F + y)^\bbot $$
\end{enumerate}
\end{definition}
Moreover, in order to talk about orthomatroids \emph{up to isomorphism}, we will say that two orthomatroids $(E_1, \bot_1)$ and $(E_2, \bot_2)$ are \emph{orthoisomorphic} if there exists a bijection $\varphi : E_1 \rightarrow E_2$ such~that
$$ \fall {x, y \in E_1} x \ortho_1 y \iff \varphi(x) \ortho_2 \varphi(y). $$

We believe that orthomatroids provide a reasonable answer to the initial objective of formalizing a notion of orthogonality-based dimension. Here, the Exchange Property ensures that a correct notion of dimension can be defined: given a  independent subset $I$, every independent subset $J$ verifying $I^\bbot = J^\bbot$ has the same cardinality as $I$ (which we call the rank of the orthomatroid). Moreover, considering orthoindependent subsets (which can be seen some sort of ``preferred'' independent subsets), the Straightening Property ensures that every closed subset admits an orthobasis and even that every orthoindependent subset can be completed into an orthobasis.

\ 

In the next section, we will study some properties of the lattices associated to orthomatroids, and show their relationship with propositional systems.

\section{The Lattice associated to an Orthomatroid}

\begin{definition}
Given an orthomatroid $M = (E, \bot)$, we define the lattice $\mc L(M)$ associated to $M$ as the set $\{ F^\bbot \mathbin\vert F \in \wp(E) \}$ of its closed subsets ordered by inclusion.
\end{definition}

Since $\mc L(M)$ is defined as the set of closed elements of a closure operator, it is a complete lattice, with operations
$$ P \wedge Q = P \cap Q \qquad P \vee Q = (P^\bot \cap Q^\bot)^\bot = (P \cup Q)^\bbot. $$
It is also clearly atomistic (meaning that every element is the join of its atoms), and it is an ortholattice with $\ \cdot\ ^\bot$ as orthocomplementation.

\medskip

We now present two important properties that are verified by $\mc L(M)$.

\begin{proposition}[Orthomodularity]
The ortholattice $\mc L(M)$ is orthomodular, that is it verifies
$$ \fall {P, Q \in \mc L(M)} P \leq Q \implies P = Q \wedge (P \vee Q^\bot). $$
\end{proposition}
\begin{proof}
Let $P$ and $Q$ be in $\mc L(M)$ such that $P \leq Q$. Clearly, $P \leq Q \wedge (P \vee Q^\bot)$. Conversely, let $x$ be in $ Q \wedge (P \vee Q^\bot)$ and suppose that $x \not \in P$. One can then define $y \in P^\bot$ such that $(P + y)^\bbot = (P + x)^\bbot$. In particular, $y \in Q$ so that $y \in Q \wedge P^\bot$ and hence $y \in (Q \wedge P^\bot) \vee Q^\bot = \bigl(Q \wedge (P \vee Q^\bot)\bigr)^\bot$. Since $x \in Q \wedge (P \vee Q^\bot)$, this implies that $y \ortho x$. But then, from $y \in P^\bot$ and $y \ortho x$, one can deduce that $y \in (P + x)^\bot$ which is absurd since $y \in (P + x)^\bbot$. As a consequence, we have shown that $\fall {x \in E} x \in Q \wedge (P \vee Q^\bot) \implies x \in P$ and finally that $P = Q \wedge (P \vee Q^\bot)$.
\end{proof}

\begin{proposition}[Atom-covering]
The ortholattice $\mc L(M)$ verifies the atom-cover\-ing property, that is for all $F \in \wp(E)$ and $x \in E$, if $x \not \in F^\bbot$, then $(F + x)^\bbot$ covers $F^\bbot$.
\end{proposition}
\begin{proof}
For $F \in \wp(E)$ and $x \in E \setminus F^\bbot$, let $G$ be such that $F^\bbot < G^\bbot \leq (F + x)^\bbot$. We have to show that $G^\bbot = (F + x)^\bbot$. But since $F^\bbot < G^\bbot$, one can define $y \in G^\bbot \setminus F^\bbot$. We then~have
$$ (F + y)^\bbot \leq G^\bbot \leq (F + x)^\bbot. $$
But then, $y \in (F + x)^\bbot \setminus F^\bbot$ which implies that $(F + x)^\bbot = (F + y)^\bbot$ and finally that $G^\bbot = (F + x)^\bbot$.
\end{proof}

This shows that if $M$ is an orthomatroid, then its associated lattice $\mc L(M)$ is a complete atomistic orthomodular lattice that satisfies the covering law or, following Piron's terminology \cite{piron64,Piron76Book}, $\mc L(M)$ is a \emph{propositional system}:

\begin{theorem}
The lattice $\mc L(M)$ associated to any orthomatroid $M$ is a propositional system.
\end{theorem}

\ 

Conversely, given a propositional system $S$, we define the associated orthoset $\mc O(S)$ made of the atoms of $S$ with orthogonality relation $ p \ortho q \iff p \leq q^\bot$. Moreover, for $F \in S$, let $\mathrm{At}(F)$ denote the set of atoms contained in $F$:
$$ \mathrm{At}(F) = \{ a \in \mc O(S) \mathbin\vert a \leq F \}. $$

\begin{proposition}
If $S$ is a propositional system, then $\mc O(S)$ is an orthomatroid.
\end{proposition}
\begin{proof} Let us first consider the exchange property, and let $F \in S$, and $x$ and $y$ be two atoms of $S$ such that $x \in \mathrm{At}(F \vee y) \setminus \mathrm{At}(F)$ (or, written in a lattice way, $x \leq F \vee y$ and $x \not \leq F$). Since $y$ is an atom not contained in $F$, $F \vee y$ covers $F$. But then $F < F \vee x \leq F \vee y$ so that $F \vee x = F \vee y$.

Now, regarding the straightening property, if $x \not \in \mathrm{At}(F)$, then considering orthomodularity, one~has
$$ F \vee x = F \vee \bigl((F \vee x) \wedge F^\bot\bigr).$$
A first consequence of this equality is that $(F \vee x) \wedge F^\bot$ contains at least one atom. Let $y$ be such an atom (so that $y \in F^\bot$). Then, since $F \vee x$ covers $F$, it follows from the exchange property that $F \vee x = F \vee y$ or, equivalently that $x \in F \vee y$.
\end{proof}

\begin{proposition}
If $S$ is a propositional system, then $\mc L(\mc O(S))$ is ortho-isomorphic (as a propositional system) to~$S$.
\end{proposition}
\begin{proof}
For the sake of clarity, let ${\ \cdot\ }^{\bot_S}$ denote the orthogonal of $S$ and ${\ \cdot\ }^{\bot_M}$ the orthogonal of orthomatroid $\mc O(S)$.  
For every subset $F$ of $\mc O(S)$ (or, equivalently, for every set of atoms of $S$), one~has
$$F^{\bot_M} = \{ a \in \mc O(S) \mathbin\vert \fall {b \in F} a \ortho b \} = \bigl\{ a \in \mc O(S) \bigm\vert a \ortho \hbox{$\bigvee$} F \bigr\} = \At \Bigl( \bigl(\hbox{$\bigvee$} F\bigr)^{\bot_S}\Bigr), $$
so~that $ F^{\bbot_M} =  \At \Bigl( \bigl(\hbox{$\bigvee$} F\bigr)^{\bbot_S}\Bigr) $ and 
$$\mc L(\mc O(S)) = \bigl\{ F^{\bbot_M} \bigm\vert F \subseteq \mc O(S) \bigr\} = \biggl\{ \At \Bigl( \bigl(\hbox{$\bigvee$} F\bigr)^{\bbot_S} \Bigr) \biggm\vert F \subseteq \mc O(S) \biggr\}. $$ 
Finally, the atomisticity of $S$ implies that
$$ \mc L(\mc O(S)) = \{ \mathrm{At}(F) \mathbin\vert F \in S \}. $$
The rest of proof follows directly from this equality. 
\end{proof}

From these results, it is clear that the lattice associated to any orthomatroid is a propositional system, but also, conversely, that any propositional system $S$ is the lattice associated to an orthomatroid, namely $\mc O(S)$. This illustrates the fact that orthomatroids do exactly capture the structure of the set of atoms of a propositional system or, stated the other way, that propositional systems are exactly the lattices associated to orthomatroids (which were defined from a discussion on dimension and orthogonality).

Moreover, any orthomatroid of the form $\mc O(S)$ for $S$ a propositional system is \emph{simple}: for every atom $x \in \mc O(S)$, one has $\{x\}^\bbot = \{x\}$. This means that there is a bijection between $\mc O(S)$ and $\At(S)$. It is well known in matroid theory that every matroid can be \emph{simplified}, i.e. transformed into a simple matroid having the same associated lattice. Here, $\mc O(\mc L(M))$ is the simplification (up to orthoisomorphism) of an orthomatroid $M$. As a consequence, in the following, we will restrict ourselves, without loss of generality, to simple orthomatroids.

\section{A Representation Theorem for Orthomatroids}

Having studied the close relationship between orthomatroids and propositional system, we will now adapt Piron's Representation Theorem to orthomatroids. This theorem can be stated~as:

\begin{theorem}[Piron’s Representation Theorem]
Every irreducible propositional system of rank at least 4 is ortho-isomorphic to the lattice of (biorthogonally) closed subspaces of a generalized Hilbert space.
\end{theorem}

Let us recall that a generalized Hilbert space $(\mc H, \mb K, \,\cdot\,^\star, \bra \,\cdot\, \vert \,\cdot\, \ket)$ consists in a vector space $\mc H$ over a field $\mb K$ with an involution anti-automorphism $\,\cdot\,^\star~: \alpha \in \mb K \mapsto \alpha^\star$ and an orthomodular Hermitian form $\bra \,\cdot\, \vert \,\cdot\, \ket~: \mc H \times \mc H \rightarrow \mb K$ satisfying
\begin{align*}
\fall {x, y, z \in \mc H} \fall {\lambda \in \mb K} & \bra \lambda x + y \vert z \ket = \lambda \bra x \vert z \ket + \bra y \vert z \ket, \\
\fall {x, y \in \mc H} & \bra x \vert y \ket = \bra y \vert x \ket^\star, \\
\fall {S \in \wp(\mc H)} & S^\bot \oplus S^\bbot = \mc H
\end{align*}
where $S^\bot = \bigl\{x \in \mc H \bigm\vert \fall {y \in S} \bra x | y \ket = 0\bigr\}$. We invite the reader to consult \cite{Stubbe:2007} for more informations.

\ 

It is well known that the set of bi-orthogonally closed subsets of a generalized Hilbert space forms a propositional system. In terms of orthomatroids, this corresponds to the fact that the set $A(\mc H)$ of vector rays of a generalized Hilbert space forms a (simple) orthomatroid with orthogonality relation
$$ \mb K x \ortho_{\mc H} \mb K y \iff \bra x \vert y \ket = 0. $$

\medskip

In order to express Piron's Representation Theorem in terms of orthomatroids, we need to define the notion of irreducibility. Following \cite{Stubbe:2007} again, given a propositional system $S$, the binary relation $\sim$ defined on $\At(S)$~by
$$ \fall {x, y \in \At(S)} x \sim y \iff \bigl(x \neq y \implies \fexist {z \in \At(S)\setminus \{x, y\}} z \leq x \vee y\bigr) $$
is an equivalence relation on $\At(S)$. The equivalence classes of $\At(S)$ are then called irreducible components of $\At(S)$, and $S$ is said to be irreducible if it has a single irreducible component. 

Given a simple orthomatroid $(E, \bot)$, the previous equivalence relation can be reexpressed as
$$ \fall {x, y \in E} x \sim y \iff \Card \{x, y\}^\bbot \neq 2. $$
If $\bigl\{(E_i, \bot_i)\bigr\}_{i \in \mc I}$ denotes the irreducible components of $E$ by  (with $\bot_i$ being the restriction of $\bot$ to $E_i$), then $(E, \bot)$ can be seen as the disjoint union of its irreducible components (up to orthoisomorphism):
\begin{gather*}
E = \biguplus_{i \in \mc I} E_i = \bigcup_{i \in \mc I} \bigl\{(i,x) \bigm\vert x \in E_i \bigr\} \\ (i,x) \ortho (j, y) \iff \bigl(i \neq j \cor (i = j \cand x \ortho_i y)\bigr)
\end{gather*}

We are now able to express Piron's Representation Theorem in terms of orthomatroids, and we finally obtain the following representation theorem:

\begin{theorem}[Representation of Orthomatroids]
Every simple and irreducible orthomatroid $(E, \bot)$ of rank at least 4 is ortho\-isomorphic to the orthomatroid $(A(\mc H), \bot_{\mc H})$ associated to a generalized Hilbert space $(\mc H, \mb K, \,\cdot\,^\star, \bra \,\cdot\, \vert \,\cdot\, \ket)$.
\end{theorem}

\section{Conclusion}

In this article, we have presented a formalism based on the idea that dimension can, at a very primitive level, be defined as the common cardinality of a maximal collections of mutually orthogonal elements (which, for instance, can be seen as spatial directions). Using a generic definition of orthogonality, together with the theory of matroids which provides a algebraic structure permitting to define a convenient notion of dimension, and a last requirement regarding the existence of orthobases, we have obtained what we call \emph{orthomatroids} which provide a general framework for dealing with dimension in an orthogonality-based context. We then have shown that orthomatroids do actually exactly capture the structure of propositional systems (or, more precisely, of the set of atoms $\At(S)$ of a propositional system $S$, which is sufficient for reconstructing $S$) with the consequence that, in high enough dimension, irreducible matroids can be represented by generalized Hilbert spaces.

\ 

We insist on the fact that in this approach, the use of generalized Hilbert lattices (or, at least, of propositional systems) is entirely derived from simple and generic geometric assumptions. As a result, this suggests that instead of seeing the use of generalized Hilbert spaces (or, again, of propositional systems) in quantum physics as puzzling, it should in the contrary be seen as the most general (if not natural) way to model situations where dimension can be defined in terms of orthogonality, such as in quantum mechanics.

\bibliographystyle{abbrv}

\end{document}